\documentclass[12pt]{amsart}

\usepackage{amsmath} 
\usepackage{amsfonts}
\usepackage{amssymb}
\usepackage{amsthm}
\usepackage{latexsym}
\usepackage{color}
\usepackage{enumerate}
\usepackage{mathrsfs}
\usepackage{setspace}
\usepackage{tikz}
\usepackage{mathtools}
\usepackage{amsbsy}
\usetikzlibrary{matrix}
\usepackage[bookmarksnumbered,colorlinks]{hyperref}

\usepackage{enumerate}
\def\br#1\er{\textcolor{red}{#1}} %
\newcommand{\N}{{\mathbb N}}

\newcommand{\cambios}[1]{{\color{black} #1}} 

\newtheorem{thm}{Theorem}[section]
\newtheorem{prop}[thm]{Proposition}
\newtheorem{lemma}[thm]{Lemma}
\newtheorem{cor}[thm]{Corollary}

\newtheorem{defi}[thm]{Definition}

\newtheorem{rem}[thm]{Remark}

\def\LLS{(X,d,\ll,\le, \tau)}
\def\LLSY{(Y,\widetilde{d},\widetilde{\ll},\widetilde{\leq},\widetilde{\tau})}

\begin{document}

\title[On the causal hierarchy of Lorentzian length spaces]{On the causal hierarchy of  \\  Lorentzian length spaces}

\author[{Ak\'e Hau}]{Luis {Ak\'e Hau}}
\address{Facultad de Matem\'aticas. Universidad Aut\'onoma de Yucat\'an. Anillo Perif\'erico, Tablaje 13615, M\'erida, Mexico.}
\email{luisakehau@gmail.com}

\author[{Cabrera Pacheco}]{Armando J. {Cabrera Pacheco}}
\address{Fachbereich Mathematik, Universit\"at T\"ubingen, Auf der Morgenstelle 10, 72076. T\"ubingen, Germany.}
\email{cabrera@math.uni-tuebingen.de}

\author[Solis]{Didier A. Solis}
\address{Facultad de Matem\'aticas. Universidad Aut\'onoma de Yucat\'an. Anillo Perif\'erico, Tablaje 13615, M\'erida, Mexico.}
\email{didier.solis@correo.uady.mx }

\begin{abstract}
In this work we provide the full description of the upper levels of the classical causal ladder for spacetimes in the context of Lorenztian length spaces, thus establishing the hierarchy between them. We also show that global hyperbolicity, causal simplicity, causal continuity, stable causality and strong causality are preserved under distance homothetic maps.
\end{abstract}

\maketitle

\section{Introduction}
\thispagestyle{empty}

In the last few years there has been a great interest for exploring solutions to the Einstein Field Equations with low regularity. Boosted by the recent findings regarding the detection of gravitational waves \cite{LIGO} and finer observation of black holes \cite{BH1}, numerous attempts to model evolution of Einstein Field Equations with non-regular \cambios{initial} data have sparkled. Let us recall that the presence of matter in a relativistic model may lead to situations in which non-smooth or even discontinuous data might be taken into consideration. Such examples include for instance black hole merging or gravitational waves \cite{HE}. As a consequence, many different alternatives to the classical framework of smooth spacetimes (that is, time oriented smooth Lorentzian manifolds) have sprung in order to deal with specific situations. For instance, spacetimes with $C^0$ metrics \cite{CG,GL,GLS,SW}, singularity theorems in $C^{1,1}$ spacetimes \cite{GGKS,KSV} and the causal set approach to quantum gravity  \cite{ABP,BLMS,BS,SB2}, just to mention a few. 

On the other hand, in the realm of Riemannian geometry the theory of (geo\-de\-sic) length spaces was developed in the 1970's  to deal with non-smooth situations. In essence, a length space is a metric space where the distance between two points can be recovered by approximating the length of curves joining them   ---see the standard references \cite{BH,burago,Searcoid} for a detailed account on such spaces. In the particular case of spaces of bounded curvature,  numerous key results suggest that this is the right framework to deal with limits and convergence. A notable example in this regard is the celebrated Precompactness Theorem which establishes that the Gromov--Hausdorff limit of a sequence of Riemannian manifolds of sectional curvature bounded from below and uniformly bounded diameter is a (Riemannian) length space with lower curvature bounds, termed an Alexandrov space \cite{Gromov1,Gromov2}.

Inspired by the classical theory of Riemannian length spaces, Kunzinger and S\"amann developed the notion of Lorentzian length spaces. In their seminal work \cite{KS}, detailed constructions are carried out in order to recover most of the standard properties of (smooth) spacetimes. In particular, they established a basic causal ladder for Lorentzian length spaces, although they did not include some of the upper levels. 

Furthermore, there is not a clear analog for the Lorentzian metric in the context of Lorentzian length spaces. Since the causal structure of a smooth spacetime is a conformal invariant, the natural question arises as to explore the kind of maps between Lorentzian length spaces that preserve causal structure.

In this work we tackle the issues related to causality described in the previous paragraphs; thus providing definitions and establishing the hierarchy of the upper levels of the causal ladder. In doing so, we also improve the notion of local causality ---see Definition 3.4 in \cite{KS}--- in order to avoid inconsistencies with the classical causal ladder of smooth spacetimes.

Moreover we describe a class of transformations that preserve these upper levels. In this way we complement the work of Kunzinger and S\"{a}mann and provide a key step in the effort of developing a causal theory for Lorentzian length spaces.

This paper is organized as follows: in Section \ref{sec:preliminares} we introduce the basic theory of Lorentzian length spaces as developed in \cite{KS}. In Section \ref{sec:CH} we introduce three of the upper levels in the causal ladder,  namely, causally simple, causally continuous and stably causal \cambios{Lorentzian length spaces}; moreover, we establish the relations between them.  Finally, in Section \ref{sec:hom} we transfer the notion of homothetic distance maps to Lorentzian length spaces and show that all the levels of the causal ladder from strong causality up are preserved under these mappings.

\section{Preliminaries}\label{sec:preliminares}

In this section we establish the notation and describe basic results pertaining to Lorentzian length spaces, closely following the theory and notation developed in \cite{KS}. First, we extract the relevant properties that serve as foundations for the causal structure of a smooth spacetime. Such axiomatic approach has been pursued in different ways (refer for instance to \cite{carter} in the context of etiological spaces or \cite{BLMS} in the theory of causets).

\begin{defi}
A \emph{causal space} $(X,\ll ,\le )$ is a set $X$ endowed with two relations $\ll$ and $\le$ that satisfy
\begin{enumerate}
\item $x\ll y$ and $y\ll z$ $\Rightarrow$ $x\ll z$.
\item $x\le y$ and $y\le z$ $\Rightarrow$ $x\le z$.
\item $x\le x$.
\item $x\ll y\Rightarrow x\le y$. 
\end{enumerate}

Moreover, we will write $p<q$ if $p \leq q$ and $p \neq q$.   
\end{defi}

The above relations give rise to the standard building blocks of the causal structure \cambios{in the standard way}.
\begin{defi}
Let $(X,\ll ,\le )$ be a causal space. The \emph{chronological future} and \emph{chronological past} of $x\in X$ are the sets defined by
\begin{eqnarray*}
I^+(x) &=& \{ y\in X\mid x\ll y\} ,\\
I^-(x) &=&\{y\in X\mid y\ll x\},
\end{eqnarray*}
respectively. In a similar fashion, the \emph{causal future} and \emph{causal past} of $x$ are
\begin{eqnarray*}
J^+(x) &=& \{ y\in X\mid x\le y\} ,\\
J^-(x) &=&\{y\in X\mid y\le x\}.
\end{eqnarray*} 
\end{defi}

Sometimes it will be useful to write the relations $\ll$ and $\le$ as subsets of $X \times X$. For example,
\begin{equation*}
J^+=\{ (x,y) \in X \times X \mid x \leq y \}.
\end{equation*}

As in the case of smooth spacetimes, the chronological sets give rise to a topology in a causal space.
\begin{defi}
Let $(X,\ll ,\le )$ be a causal space. The \emph{Alexandrov topology} $\mathcal{A}$ in $X$ is the topology having as a  subbase the collection of chronological diamonds $I^+(x)\cap I^-(y)$.
\end{defi}

In addition to the causal structure, the second main ingredient to be considered in the definition of a Lorentzian length space is an analog of the Lorentzian distance function (time separation) in the smooth context. 
\begin{defi}
Let $(X, \ll ,\le)$ be a causal space and $d$ a metric distance in $X$. A function $\tau : X\times X\to [0,\infty ]$ that satisfies
\begin{enumerate}
\item $\tau$ is lower semicontinuous with respect to the topology induced by $d$.
\item $\tau (x,z)\ge \tau (x,y)+\tau (y,z)$ for all $x\le y\le z$.
\item $\tau (x,y)=0$ if $x\not\le y$.
\item $\tau (x,y)>0 \Leftrightarrow x\ll y$.
\end{enumerate}
is called a \emph{time separation} for $(X,\ll ,\le )$. 
\end{defi}

\begin{defi}
A \emph{Lorentzian pre-length space} $\LLS$ consists of a causal space $(X,\ll ,\le )$ endowed with a time separation $\tau$. 
\end{defi}

We use $d$ to furnish $X$ with a (metric) topology $\mathcal{D}$, henceforth any topological notion will be considered with respect to this topology, unless otherwise explicitly stated.

\begin{rem}
Recall that a smooth spacetime $(M,g)$ is always metrizable, so in the smooth scenario we can regard $d=d_M$ as the metric that induces the manifold topology, and hence $(M, d_M, \ll_g ,\le_g , \tau_g )$ is a Lorentzian pre-length space, where $\ll_g$, $\le_g$  are the usual chronological and causal relations and $\tau_g$ the standard Lorentzian distance function. 
\end{rem}

In a Lorentzian pre-length space there is just enough structure to deduce some of the basic causal properties we might expect for a non-smooth version of a spacetime. Here we mention a few.

\begin{lemma}[Push up, \cite{KS} Lemma 2.10]\label{lema:push}
Let $\LLS$ be a Lorentzian pre-length space and let $x\le y\ll z$ or $x\ll y\le z$. Then $x\ll z$.
\end{lemma}

\begin{prop}[\cite{KS} Lemma 2.12 and Prop. 2.13]\label{prop:Iopen}
Let $\LLS$ be a Lorentzian pre-length space. The chronological sets $I^{\pm}(x)$ are open (in the topology induced by $d$). Furthermore, the relation $\ll$ is open (as a subset of $X\times X$).
\end{prop}

Lorentzian pre-length spaces owe its name to the fact that there is a distance inspired notion of length for causal curves.

\begin{defi}\label{defi:causalcurve}
Let $\LLS$ be a Lorentzian pre-length space. A \emph{future directed causal (timelike)} curve is a non-constant  Lipschitz continuous map $\gamma : I\subset \mathbb{R}\to (X,d)$ with the property that for all $s<t$ we have $\gamma (s)\le \gamma (t)$ ($\gamma (s)\ll \gamma (t)$). A future directed causal curve is called \emph{future null} if no two points on the curve are related with respect to $\ll$. Past directed causal, timelike and null curves are defined analogously.
\end{defi}

\begin{rem}
This notion of causality does not extend in general the common notion of causality for smooth spacetimes. For instance, consider the quotient of the two-dimensional Lorentz-Minkowski spacetime $\mathbb{L}^2$ under the identification $(x,t)\simeq (x,t+1)$. The resulting smooth spacetime $(M,g)$ is totally vicious, so every pair of points can be connected by a smooth future directed timelike curve. Thus, according to Definition \ref{defi:causalcurve}, the Lorentzian pre-length space $(M, d_M, \ll_g ,\le_g , \tau_g )$ does not have any null curve. On the other hand, the curve $\gamma :\mathbb{R} \to M$, $\gamma (s)=(s,1/2)$ though not causal in $(M,g)$, it is timelike in $(M, d_M, \ll_g ,\le_g , \tau_g )$. Nevertheless, the above discrepancy is due to the poor causal properties of the spacetime $(M,g)$. In fact, for strongly causal spacetimes the two notions of causality agree (refer to Lemma 2.21 in \cite{KS}).
\end{rem}

\begin{defi}
Let $\gamma :[a,b]\to X$ be a future directed causal curve in the Lorentzian pre-length space $\LLS$. We define the $\tau$\emph{-length of } $\gamma$ by
\begin{equation*}
L_\tau (\gamma )=\inf \left\{ \sum_{i=0}^{N-1}\tau (\gamma (t_i), \gamma (t_{i+1}))\right\}
\end{equation*}
where the infimum runs over all partitions $a=t_0<t_1<t_2\ldots <t_{N-1}<t_N=b$  of the interval $[a,b]$. In case $L_\tau (\gamma\vert_{s}^t )>0$ for all $a\le s<t\le b$ we say $\gamma$ is \emph{rectifiable}.
\end{defi}

As a direct consequence of the definitions, we have the following result that mimic the properties of a length structure in Riemannian length spaces \cite{burago}.

\begin{prop}[\cite{KS} Lemmas 2.25, 2.27 and 2.30]
 Let $\gamma :[a,b]\to X$ be a future directed causal curve in the Lorentzian pre-length space $\LLS$. Then
 \begin{enumerate}
 \item $L_\tau (\gamma\vert_{a}^b)= L_\tau (\gamma\vert_{a}^s )+L_\tau (\gamma\vert_{s}^b)$ for all $a<s<b$.
 \item $L_\tau (\gamma)$ is invariant under reparametrizations.
 \item If $\gamma$ is rectifiable then it is timelike.
 \end{enumerate}
\end{prop}

At this point it is worthwhile noticing that the causal (chronological) relation is not a priori related to causal (chronological) connectivity. Thus the following definition is needed.

\begin{defi}
A Lorentzian pre-length space $(M,d,\ll ,\le , \tau )$ is called \emph{causally path connected} if for every pair of points with $x\le y$ ($x\ll y$) there exists a future directed causal (timelike) curve joining them.
\end{defi}

We now turn our attention to another of the key features of causality in the smooth settings: the fact that geometry and causality when restricted to a (convex) normal neighborhood $U$  are much less intricate that in the large. This fact in turn relies on two important properties of convex neighborhoods: (i) the local time separation $\tau_U$ is smooth, and (ii) the local causal relation $\le_U$ is closed (refer to \cite{BE} Lemma 4.26). We can recover both properties in the context of Lorentzian length spaces according to the following definitions.

\begin{defi}
\label{localizing}
A Lorentzian pre-length space $\LLS$ is \emph{localizable} if for each $x\in X$ there is a neighborhood $\Omega_x$ ---termed \emph{localizing neighborhood}--- that satisfies
\begin{enumerate}
\item All causal curves contained in $\Omega_x$ have uniformly bounded $d$-length.
\item There is a continuous map $\omega_x:\Omega_x\times \Omega_x\to [0,\infty )$ such that
\begin{enumerate}
\item $(\Omega_x,d\vert_{\Omega_x\times \Omega_x},\ll_{\Omega_x},\le_{\Omega_x},\omega_x)$ is a Lorentzian pre-length space.
\item For every $y\in \Omega_x$, $I^{\pm}(y)\cap\Omega_x\neq \emptyset$ 
\end{enumerate}
\item For all $p,q\in\Omega_x$ with $p < q$ there exists a future causal curve $\gamma_{p,q}$ from $p$ to $q$  with
\begin{enumerate}
\item $L_\tau (\gamma_{p,q})\ge L_\tau (\gamma )$ for all future causal curves from $p$ to $q$ with $\gamma\subset \Omega_x$.
\item $L_\tau (\gamma_{p,q})=\omega_x(p,q)$.
\end{enumerate}
\end{enumerate}
If in addition for $p,q\in\Omega_x$ with $p\ll q$
\begin{enumerate}
\item The curve $\gamma_{p,q}$ is future timelike.
\item $L_\tau (\gamma_{p,q})>   L_\tau (\gamma )$ for all future causal curves from $p$ to $q$  contained in $\Omega_x$ and having a null segment.
\end{enumerate}
we say that $\LLS$ is \emph{regularly localizable}. If every point $x\in X$ have a basis of localizing neighborhoods we say $\LLS$ is \emph{strongly localizable}.  If every  point $x\in X$ has a basis of regularly localizing neighborhoods we say $\LLS$ is \emph{SR-localizable}.
\end{defi}

In \cite{KS}, a neighborhood $U$ of a Lorentzian pre-length space $\LLS$ is called causally closed if for any sequences $\{ p_n\}$, $\{q_n\}$ in $U$ with $p_n\le q_n$, converging to $p,q\in\overline{U}$, respectively,  it follows that $p\le q$. As it turns out, this notion leads to inconsistencies in the classical causal ladder of smooth spacetimes. More precisely, it is possible to show that with the definition given in \cite{KS}, the causal condition is equivalent to strong causality, \cambios{a statement} which is known to be false for smooth spacetimes. Therefore, for the sake of consistency we re-introduce the notion of closed causality in a way that resembles more closely the smooth scenario.

\begin{defi}
	\label{loccauc}
In a causally path connected Lorentzian pre-length space $\LLS$ we can define a causal relation $\leq_{U}$ in any open subset $U$ as follows:
\begin{itemize}
\item $p \leq_{U} q$ $\Leftrightarrow$ there exists a future directed causal curve $\gamma:[a,b] \rightarrow X$, with $\gamma(a)=p$ and $\gamma(b)=q$, such that  $\gamma([a,b]) \subset U$.
\end{itemize} 
\end{defi} 

It is important to note that in a causally path connected Lorentzian pre-length space and for any open subset $U$, we will 
have that $p \leq_{U} q$ implies $p \leq q$ in $X$. Now, we define our notion of \cambios{local causal closure}.

\begin{defi} \label{loccaucl}
A neighborhood $U$ of a Lorentzian pre-length space $\LLS$ is called \emph{causally closed} if for any sequences $\{ p_n\}$, $\{q_n\}$ with $p_n\le_{U} q_n$, converging to $p,q \in U$, respectively,  it follows that $p\le_{U} q$. If every point $x\in X$ has a causally closed neighborhood we say that  $\LLS$ is \emph{locally causally closed}.
\end{defi}

In the context of spacetimes with a smooth metrics $(M,g)$ the previous neighborhoods exist for every point see \cite[Lemma 14.2]{Oneill}.

\begin{defi}
A \emph{Lorentzian length space} is a causally connected, locally causally closed and localizable Lorentzian pre-length space such that the time separation function $\tau$ satisfies
\begin{equation*}
\tau (x,y) =\sup \{ L_\tau (\gamma )\ \mid \ \gamma\text{ is a future causal curve from }x\text{ to }y \} .
\end{equation*}
\end{defi}

As examples of Lorentzian length spaces we have strongly causal spacetimes with continuous and causally plain metrics \cite{CG, KS}, closed cone structures \cite{MinguzziC}, \cambios{generalized cones \cite{AGKS}} and timelike causal funnels \cite{KS}.

One of the main local properties of Lorentzian length spaces that we will use in this work is described in the following lemma.

\begin{lemma}[Sequence Lemma]\label{lema:seq}
Let $\LLS$ be a Lorentzian length space and $x\in X$. Then, there exist sequences $p_n\to x$, $q_n\to x$ such that
\begin{enumerate}
\item $p_n\ll x\ll q_n$, $\forall n\in\mathbb{N}$.
\item $p_n\ll p_{n+1}$ and $q_{n+1}\ll q_n$, $\forall n\in\mathbb{N}$.
\end{enumerate}
\end{lemma}
\begin{proof}
Let $\Omega_x$ be a localizing neighborhood of $x$. Then, there exist $p_1\in I^-(x)\cap \Omega_x$ and $q_1\in I^+(x)\cap\Omega_x$. By causal connectedness, there exist a future 
directed timelike curve $\alpha :[0,1]\to X$ from $p_1$ to $x$ and a future directed timelike curve $\beta :[0,1]\to X$ from $x$ to $q_1$. Thus $p_n=\alpha (1-1/n)$, $q_n =\beta (1/n)$ are the desired sequences. 
\end{proof}

As the first consequence of the Sequence Lemma, we have the following characterization of the closure of future or past sets in a Lorentzian length space (\cambios{compare to} \cite[Proposition 3.3]{Penrose}).
\begin{prop}\label{prop:clI} 
Let $(X,\ll ,\le )$ be a Lorentzian length space. Then 
\begin{eqnarray*}
\overline{I^+(x)}&=&\{y\in X\, \mid\, I^+(y)\subset I^+(x) \},\\
\overline{I^-(x)}&=&\{y\in X\, \mid\, I^-(y)\subset I^-(x) \}.
\end{eqnarray*}
\end{prop}
\begin{proof}
Let $y\in \overline{I^+(x)}$, thus there is a sequence $y_n\to y$, $y_n\in I^+(x)\ \forall n\in\mathbb{N}$. Let $z\in I^+(y)$, ---this $z$ exists by the Sequence Lemma and the existence of localizing neighborhoods--- then $y\in I^-(z)$. From Proposition \ref{prop:Iopen} we have that $y_n\in I^-(z)$ for sufficiently large $n$. Thus $x\ll y_n\ll z$ for all such values of $n$. That is, $z\in I^+(x)$. Hence $I^+(y)\subset I^+(x)$. Conversely, let $I^+(y)\subset I^+(x)$ and consider a sequence $y_n\to y$, $y_n\in I^+(y),\ \forall n\in\mathbb{N}$ in virtue of Lemma \ref{lema:seq}. Thus, $y_n\in I^+(x)$ and hence $y_n\in \overline{I^+(x)}$ proving that $\overline{I^+(x)}=\{y\in X\, \mid\, I^+(y)\subset I^+(x) \}$. The characterization of $\overline{I^-(x)}$ is done in an analogous way.
\end{proof}

\section{Causal hierarchy}\label{sec:CH}

In \cite{KS} Kunzinger and S\"amann defined the causal hierarchy for both causal spaces and Lorentzian length spaces. For the latter they considered the (partial) causal ladder: 

\medskip

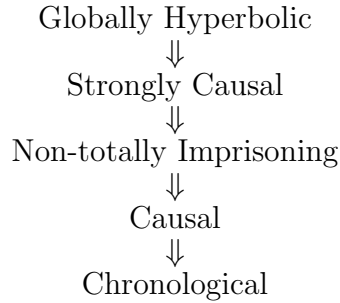
\begin{center}
\begin{tabular}{c}
Globally Hyperbolic\\
$\Downarrow$\\
Strongly Causal\\
$\Downarrow$\\
Non-totally Imprisoning\\
$\Downarrow$\\
Causal\\
$\Downarrow$\\
Chronological
\end{tabular}
\end{center}
with the following definitions.

\begin{defi}
A Lorentzian length space $\LLS$  is
\begin{enumerate}
\item \emph{Chronological} if $x\not\ll x$ for all $x\in X$.
\item \emph{Causal} if $x\le y$ and $y\le x$ implies $x=y$.
\item \emph{Non-totally imprisoning} if for every compact set $K\subset X$ there is a $C\ge 0$ such that the $d$-arclength of all causal curves contained in $K$ is bounded by $C$.
\item \emph{Strongly causal} if the Alexandrov topology $\mathcal{A}$ coincides with the topology $\mathcal{D}$ induced by $d$.
\item \emph{Globally hyperbolic} if it is non-totally imprisoning and the causal diamonds $J^+(x)\cap J^-(y)$ are compact sets.
\end{enumerate}
\end{defi}

\begin{rem} We would like to remark the following: 
\begin{enumerate}
\item The notions (1)-(4) are equivalent to the corresponding usual notions in smooth spacetimes (see \cite[Lemma 3.3, Corollary 3.15 and Theorem 3.26]{KS}). 
 \item For smooth spacetimes, the notion global hyperbolicity \cambios{does not require non-imprisonment \cite{HM,MS}.} None\-theless, for spacetimes with $C^0$ metrics, \cambios{this condition}  is required \cite{Saemann}.
\end{enumerate}
\end{rem}

Notice that in the version of the causal ladder presented in \cite{KS} the upper levels corresponding to causal simplicity, causal continuity and stable causality are not considered. Although these levels are usually defined for smooth spacetimes in terms of the regularity of time functions, all these levels admit formulations solely in terms of the causal structure. Before we do so, we introduce a couple more concepts of causal theory. Thus we follow this latter approach and consider the following definitions for the upper levels of the causal ladder for Lorentzian length spaces.

\begin{defi} \label{defi:distinguishing}
A Lorentzian length space $\LLS$ is \emph{future (past) distinguishing} if $I^+(x)=I^-(y)$ ($I^-(x)=I^-(y)$) implies $x=y$. If $\LLS$ is both future and past distinguishing, then it is called \emph{distinguishing}.
\end{defi}

\begin{defi}
A Lorentzian length space $\LLS$ is \emph{past reflective} if the condition
$ I^{+}(x) \subset I^{+}(y) \Longrightarrow I^{-}(y) \subset I^{-}(x)$  holds. Similarly, it is \emph{future reflective} if $ I^{-}(y) \subset I^{-}(x) \Longrightarrow I^{+}(x) \subset I^{+}(y)$. A future and past reflective Lorentzian length space is called \emph{reflective.}
\end{defi}

\begin{rem}\label{rem:refeq}
Notice that in virtue of Proposition \ref{prop:clI} future reflectivity is equivalent to the condition $y \in \overline{I^{+}(x)} \Longrightarrow x \in \overline{I^{-}(y)}$ and likewise, past reflectivity is equivalent to $y \in \overline{I^{-}(x)} \Longrightarrow x \in \overline{I^{+}(y)}$ (see \cite[Lemma 4.32, (i) and (ii)]{MS}).
\end{rem}

Following \cite{SW}, we make the following definition.

\begin{defi}\label{defi:K}
Let $\LLS$ be a Lorentzian length space. We define the relation $K^+\subset X\times X$ as the smallest transitive and closed\footnote{$K^+$ is closed if $\{(x_n,y_n)\}\to (x,y)$ with $\{(x_n,y_n)\}\subset K^+$ implies $(x,y)\in K^+$.} relation that includes $J^+=\{(x,y)\in X\times X\, \mid\, x\le y\}$.
\end{defi}

Now we state the upper levels of the causal ladder.

\begin{defi}
A Lorentzian length space $\LLS$ is \emph{causally simple} if it is causal and the sets $J^{\pm}(x)$ are closed for all $x\in X$. 
\end{defi}

\begin{defi}\label{defi:CC}
A Lorentzian length space $\LLS$ is \emph{causally continuous} if it is distinguishing and reflective.
\end{defi}

\begin{defi}\label{defi:SC}
A Lorentzian length space $\LLS$ is \emph{stably causal} if the relation $K^+$ is antisymmetric.
\end{defi}

\begin{rem}
For smooth spacetimes we have:
\begin{enumerate}
\item The equivalence between definition \ref{defi:CC} and the standard notion of causal continuity can be found in \cite{MS}. 
\item The definition \ref{defi:SC} corresponds to the notion fo $K$-causality introduced by Sorkin and Woolgar in \cite{SW} and studied extensively by Minguzzi, who established the equivalence between the usual notion of stable causality and $K$-causality \cite{MinguzziKC}.
\end{enumerate}
\end{rem}

As in the case of smooth spacetimes, we move on into proving the different logical implications between the above causality conditions, thus establishing the causal hierarchy for Lorentzian length spaces.

\begin{lemma}\label{lemma:cl}
Let $\LLS$  be a Lorentzian length space. Then $\overline{I^{\pm}(x)}=\overline{J^{\pm}(x)}$ for all $x \in X$.
\end{lemma}
\begin{proof}
By Proposition \ref{prop:Iopen}  we have $I^{\pm}(x) \neq \emptyset$ for all $x \in X$. Let $y\in J^+(x)$, then by the Sequence Lemma \ref{lema:seq} there exists a sequence $y_n\to y$ with $y_n\in I^+(y)$. Thus, the Push up Lemma  guarantees that $y_n\in I^+(x)$ and hence $y\in\overline{I^{+}(x)}$. Therefore $J^{+}(x) \subset \overline{I^{+}(x)}$ for all $x \in X$, and the claim follows from $I^{+}(x) \subset J^{+}(x)$. The past case is analogous.
\end{proof}

\begin{prop}\label{GHtoCS}
If $\LLS$ is a globally hyperbolic Lorentzian length space, then $\LLS$ is causally simple.
\end{prop}

\begin{proof}
Suppose that $\LLS$ is globally hyperbolic, then it is non totally imprisoning and hence causal. It only remains to show that the sets $J^{\pm}(x)$ are closed subsets for all $x \in X$. Let $x \in X$ and $y \in \overline{J^{+}(x)}$. By Lemma \ref{lemma:cl} there exists a sequence $y_{n} \in I^{+}(x)$ such that $y_{n} \rightarrow y$. For any point $y^{+} \in I^{+}(y)$ we have that $I^{-}(y^{+}) \cap I^{+}(x) \neq \emptyset$. Since $I^{-}(y^{+})$ is open and $\{y_{n}\}$ converges to $y$, we have that $y_{n} \in I^{-}(y^{+}) \cap I^{+}(x) \subset J^{-}(y^{+}) \cap J^{+}(x)$, where the latter subset is compact, so $ \overline{J^{-}(y^{+}) \cap J^{+}(x)}=J^{-}(y^{+}) \cap J^{+}(x)\subset J^{+}(x)$. Therefore, $y \in J^{+}(x)$ and this implies that $J^{+}(x)$ is closed. The past case is analogous.   
\end{proof}

\begin{prop}\label{CStoCC}
If $\LLS$ is a causally simple Lorentzian length space, then $\LLS$ is causally continuous.
\end{prop}

\begin{proof}
Suppose that $\LLS$ is causally simple. We have to prove that $\LLS$ is distinguishing and reflective. First, let us prove the distinguishing condition. By contradiction suppose that  $\LLS$ is not distinguishing, then, there exist $p,q \in X$ with $p \neq q$ and such that $I^{-}(p)=I^{-}(q)$. Take $U$ and $V$ two disjoint neighborhoods of $p$ and $q$ respectively. By the Sequence Lemma we can take a sequence  
$\{q_n\} \subset V$ with $q_n \ll q$ and $q_n \rightarrow q$. Since $I^{-}(q)=I^{-}(p)$ we have that $q_n \in I^{-}(p)$, therefore $q \in \overline{I^{-}(p)}=\overline{J^{-}(p)}=J^{-}(p)$. An analogous reasoning shows that $p \in \overline{I^{+}(q)}=\overline{J^{+}(q)}=J^{+}(q)$. Then, $p \leq q \leq p$ and $p \neq q$ and this is a contradiction since $\LLS$ is a causal Lorentzian length space.    
It only remains to prove that $\LLS$ is reflecting. Further, by Remark \ref{rem:refeq} we need to show that conditions 
\begin{equation*}
y \in \overline{I^{+}(x)} \Longrightarrow x \in \overline{I^{-}(y)}\quad \text{and}\quad y \in \overline{I^{-}(x)} \Longrightarrow x \in \overline{I^{+}(y)}
\end{equation*}
hold in our case. Let $y \in \overline{I^{+}(x)}$, then $y \in \overline{I^{+}(x)}=\overline{J^{+}(x)}=J^{+}(x)$  and so $x \in J^{-}(y)=\overline{J^{-}(y)}=\overline{I^{-}(y)}$. The proof of the second condition is analogous. Then, $\LLS$ is reflecting, and hence it is causally continuous.   
\end{proof}

\begin{prop}\label{CCtoSC}
If $\LLS$ is a causally continuous Lorentzian length space, then $\LLS$ is stably causal.
\end{prop}

\begin{proof}
We start by defining the relation $R^+\subset X\times X$ by
\begin{equation*}
R^+=\{ (x,y) \, \mid\, y\in\overline{I^+(x)}\}.
\end{equation*}
Notice that if $(x,y)\in J^+$ then $y\in J^+(x)\subset \overline{I^+(x)}$, so, $(x,y)\in R^+$ and this implies that $J^+\subset R$. Now let us prove that $R^+$ is a closed subset of $X\times X$. Take $(x,y) \in \overline{R^+}$, so, there exists a sequence $\{(x_m,y_m)\} \subset R^+$ with $(x_m,y_m) \rightarrow (x,y)$. Consider a sequence $\{q_n\} \subset I^+(y)$ with $q_n \rightarrow y$ and $q_{n+1} \ll q_{n}$.  By the convergence of $\{y_{m}\}$ to $y$ and the fact that $y_{m} \in \overline{I^{+}(x_m)}$ we have that $q_n \in I^{+}(x_m)$ for $m \geq M_{0}(n)$ for some $M_{0}(n) \in \N$. Hence $I^-(q_n)\cap I^+(x_m)$ is a non-empty open set for $m\geq M_0(n)$.
Thus, $x_{m} \ll q_{n}$ for $m \geq M_{0}(n)$ and this implies that $x \in \overline{I^{-}(q_{n})}$ since $x_{m} \rightarrow x$. Since $\LLS$ is a causally continuous Lorentzian length space we have that $x \in \overline{I^{-}(q_{n})}$ implies that $q_{n} \in  \overline{I^{+}(x)}$ (see Remark \ref{rem:refeq}) and observe that the inclusion is valid for all $n \in \N$. The convergence of $\{q_{n}\}$ to $y$ implies that $y  \in \overline{\overline{I^+(x)}}=\overline{I^{+}(x)}$, so, $(x,y) \in R^+$ and $R^+$ is closed.

Furthermore, $R^+$ is transitive. To show this, let $(x,y)\in R^+$ and $(y,z)\in R^+$. Thus $y \in \overline{I^{+}(x)}$ and $z \in \overline{I^{+}(y)}$, the characterization of the closure of future sets provides the following inclusions:
\[
I^{+}(z) \subset I^{+}(y) \subset I^{+}(x).
\]   
We obtain that $I^{+}(z) \subset I^{+}(x)$ and this is equivalent to $z \in \overline{I^{+}(x)}$, so, $(x,z) \in R^+$. As a consequence of the definition of $K^+$ we have $K^+\subset R^+$.

Finally, let $(x,y)$ and $(y,x)$ be in $R^+$. As before we have $I^+(y) \subset I^+(x)$ and $I^+(x) \subset I^{+}(y)$ from $y\in \overline{I^+(x)}$ and $x\in \overline{I^+(y)}$ respectively. Therefore, $I^{+}(x)=I^{+}(y)$ and by the distinguishing condition we obtain $x=y$, so, $R^+$ is antisymmetric and this implies that $K^+$ is antisymmetric as well.
\end{proof}

\begin{prop}\label{SCtoSC}
Let $\LLS$ be a Lorentzian length space with $(X,d)$ locally compact metric space. If $\LLS$ is stably causal then $\LLS$ is strongly causal.
\end{prop}
\begin{proof}
Assume that $\LLS$ is not strongly causal in $x\in X$. Thus there exists a $\mathcal{D}$-neighborhood $W$ of $x$ that does not contain any $\mathcal{A}$-neighborhood of $x$. Let $\Omega_x$ be a localizing neighborhood of $x$ and $U$ a causally closed neighborhood of $x$ (see Definition \ref{loccaucl}). Let us further consider $V=W\cap U\cap\Omega_x$ and $B^d_{r}(x)$ (for some small $r>0$) with compact closure contained in $V$. By the Sequence Lemma, consider sequences $p_n\to x$, $q_n\to x$ with $p_n\ll x\ll q_n$, $p_n\ll p_{n+1}$, $q_{n+1}\ll q_n$, $\forall n\in\mathbb{N}$. Since $I^+(p_n)\cap I^-(q_n)$ is an Alexandrov open set that contains $x$, it can not be a subset of $U$ and so can not be contained in $B^d_{r}(x)$ \cambios{either}. Thus, there is $w_n\in X \setminus B^d_{r}(x)$ with $p_n\ll w_n\ll q_n$. By causal conectedness, there is a future timelike curve $\alpha_n$ from $p_n$ to $w_n$. Let $y_n$ be the the first point of $\alpha_n$ that meets $\partial B^d_{r}(x)$ and observe that for these points we have $p_{n} \ll_{U} y_{n}$ for all $n$. 
Since $\partial B^d_{r}(x)$ is a compact subset, there is a subsequence $\{y_{n_k}\}$ that converges to $y\in B^d_{r}(x)$. Since $U$ is causally closed and $p_{n_k}\ll_{U} x_{n_k}$ we have $x \le_{U} y$, so, $(x,y)\in J^+ \subset K^+$. Notice further that $x \neq y$ since $y \in \partial B^d_{r}(x)$. 
Also, observe that $(y_{n_{k}},q_{n_{k}}) \in I^{+} \subset J^{+}$ and thus $(y,x) \in \overline{J^{+}}$ since $(y_{n_{k}},q_{n_{k}}) \rightarrow (y,x)$. Then, $K^+$ is not  antisymmetric since $(x,y) \in J^{+} \subset K^+$ and $(y,x) \in \overline{J^{+}} \subset K^+$ with $x \neq y$.
\end{proof}

Thus, taking into account Propositions \ref{GHtoCS} to \ref{SCtoSC} we have established the hierarchy of the upper levels of the causal ladder. We summarize our results in the following theorem.

\begin{thm}
\label{causalthm}
Let $\LLS$ be a Lorentzian length space. Then, we have the following implications:
\begin{center}
\begin{tabular}{c}
Globally Hyperbolic\\
$\Downarrow$\\
Causally Simple\\
$\Downarrow$\\
Causally Continuous\\
$\Downarrow$\\
Stably Causal\\
\end{tabular}
\end{center}
In addition, if we assume $(X,d)$ is locally compact, then 
\begin{center}
\begin{tabular}{c}
Stably causal\\
$\Downarrow$\\
Strongly Causal
\end{tabular}
\end{center}
\end{thm}

We close this section establishing characterizations of some levels of the causal ladder in terms of the Lorentzian distance function.

\begin{prop}
Let $\LLS$ be a distinguishing Lorentzian length space. If  $\tau:X \times X \rightarrow [0,\infty]$ is continuous, then $\LLS$ is causally continuous.
\end{prop}

\begin{proof}
Since $\LLS$ is distinguishing we only need to prove the reflecting condition. So, let us prove that $I^{+}(x) \subset I^{+}(y)$ implies $I^{-}(x) \subset I^{-}(y)$. Take $z \in I^{-}(p)$ and a sequence $\{y_n\} \subset I^{+}(y)$ such that $y_{n} \rightarrow y$ and $y_{n} \ll y_{n+1}$, therefore, $z \ll x \ll y_n$ since $y_n \in I^{+}(x)$ by hypothesis. By the reverse triangle inequality  we have
\[
\tau(z,y_n) \geq \tau(z,x) +\tau(x,y_{n}),
\]
hence continuity of $\tau$ gives us the following inequality:
\[
\tau(z,y) \geq \tau(z,x) +\tau(x,y) \geq \tau(z,x) > 0
\]
which is equivalent to $z \in I^{-}(y)$. Therefore, $I^{-}(x)\subset I^{-}(y)$. By an analogous reasoning we can show $I^{-}(y) \subset I^{-}(x)$ implies $I^{+}(x) \subset I^{+}(y)$, so, $\LLS$ is causally continuous. 
\end{proof}

\begin{prop}
If $\LLS$ is a strongly causal Lorentzian length space, then for each $x \in X$ there exists a neighborhood $V \ni x$ such that $\tau\mid_{V \times V}$ is continuous. 
\end{prop}

\begin{proof}
Suppose that $\LLS$ is a strongly causal Lorentzian length space. Let $x \in X$ and consider a localizing neighborhood $\Omega_x \ni x$ and its associated continuous function $\omega_{x}: \Omega_{x} \times \Omega_{x} \rightarrow [0,\infty]$. 
Then, there exists a causally convex neighborhood $V \subset \Omega_x$ by strong causality.
Note that the causal relation $\le_{V}$, as defined in Definition \ref{loccauc}, coincides with the restricted causal relation $\le\mid_{V \times V}$, furthermore, it coincides with $\le\mid_{\Omega_{x} \times \Omega_{x}}$. Observe that $\omega_{x}$ restricted to $V \times V$ is a continuous function. We will prove that $\tau \mid_{V \times V}=\omega_{x}\mid_{V \times V}$. Indeed, for any $p,q \in V \subset \Omega_{x}$ with $p<q$ there exists a future directed causal curve $\gamma_{pq}:[a,b] \rightarrow X$ with $\gamma(a)=p$, $\gamma(b)=q$, $\gamma_{pq}([a,b]) \subset \Omega_{x}$ and $L_{\tau}(\gamma_{pq})=\omega_{x}(p,q) \leq \tau(p,q)$ (see Definition \ref{localizing}). If $\omega_{x}(p,q) < \tau(p,q)=\sup\{L_{\tau}(\gamma) \mid \ \gamma\text{ is a future causal curve from }x\text{ to }y\}$, then there must exists a future directed causal curve $\sigma$ between $p$ and $q$ with $L_{\tau}(\gamma_{pq})=\omega_{x}(p,q)<L_{\tau}(\sigma)$. Note that $\sigma$ must be included in $V$ by the causal convexity of $V$, so $\sigma \subset \Omega_{x}$ and we have arrived to a contradiction to the maximality of $\gamma_{pq}$ in $\Omega_{x}$. Then, $\omega_{x}(p,q)=\tau(p,q)$. If $p=q$, we have that $\tau(p,p)=0$ since $\LLS$ is strongly causal, thus, $\omega_{x}(p,p)=0$. Therefore, $\tau\mid_{V \times V}=\omega_{x}\mid_{V \times V}$ and so $\tau$ is continuous in $V \times V$.

\end{proof}

\section{Homothetic maps and Lorentzian length spaces}\label{sec:hom}

It is well known that in the smooth scenario, causal structure ---and hence causal hierarchy--- is conformally invariant. Thus, conformal transformations play a key role in describing the causal structure of spacetime. For instance, a famous result of Zeeman states that the group of causal automorphisms of Lorentz-Minkowski spacetime is generated by conformal transformations \cite{Zeeman}.

 On the other hand, the absence of a Lorentzian metric in a Lorentzian length space poses the problem of finding the appropriate kind of transformations between them that preserve their causal structure, or at least, their causal hierarchy. In this section we show that the class of distance homothetic maps is a good candidate for the upper levels of the causal ladder.

\begin{defi}
Let $\LLS$ and $\LLSY$ be two pre-Lorentzian length spaces. A map $f:X \rightarrow Y$ is said to be \emph{distance homothetic} if there exists a constant $c>0$ such that $\widetilde{\tau}(f(p),f(q))=c \tau(p,q)$ for all $p,q \in X$. If $c=1$, then $f$ is said to be a distance preserving map.  
\end{defi}

Following closely the analog results for smooth spacetimes in \cite{BE} we establish first a couple of technical lemmas. We include the proofs here for completeness.

\begin{lemma}
\label{ftime}
Let $\LLS$ and $\LLSY$ be pre-Lorentzian length spaces, and consider a map $f:X \rightarrow Y$ which is surjective but not necessarily continuous. If $f$ is distance homothetic, then 
\begin{enumerate}
	\item $p \ll q$  if and only if $f(p) \ll f(q)$.
	\item $f(I^{\pm}(p))=I^{\pm}(f(p))$ for all $p\in X$.
	\item $f(I^{{\pm}}(p) \cap I^{\mp}(q))=I^{{\pm}}(f(p)) \cap I^{\mp}(f(q))$.
\end{enumerate} 
\end{lemma}

\begin{proof}
	
	We divide the proof in three steps:
	
	\begin{enumerate}
	\item Observe that $p \ll q$ if and only if $\tau(p,q)>0$, since $f$ is distance homothetic we have that $\widetilde{\tau}(f(p),f(q))=c \tau(p,q)>0$ and this is equivalent to $f(p) \widetilde{\ll} f(q)$.
	
	\item Let $q' \in f(I^+(p))$, then there exists $q \in I^+(p)$ with $q'=f(q)$ and by the previous item we have that $f(q)=q' \in I^+(f(p))$, so $f(I^{+}(p)) \subset I^{+}(f(p))$. To prove the other inclusion let $q' \in I^+(f(p))$ and note that since $f$ is onto there exists $q \in X$ such that $q'=f(q)$, so, $f(p) \widetilde{\ll} f(q)=q'$ and this is equivalent (by the previous item) to $p \ll q$. Therefore, $q'\in f(I^+(p))$ and so $I^{+}(f(p)) \subset f(I^{+}(p))$.
	
	\item The first item implies that for any $r$ with $p \ll r \ll q$ we will have $f(p) \widetilde{\ll} f(r) \widetilde{\ll} f(q)$, which gives the result.     	
	\end{enumerate}
\end{proof}

Recall that a pre-Lorentzian length space is strongly causal if the metric topology coincides with the Alexandrov topology, hence in this case we will have that a homothetic map is an open map.

\begin{lemma}\label{lema:43}
	Let  $\LLS$ be a strongly causal Lorentzian length space and $\LLSY$ be an arbitrary Lorentzian length space. If $f: X \rightarrow Y$ is a surjective distance homothetic map (not necessarily continuous) then $f$ is an open map and it is also injective.
\end{lemma}

\begin{proof}

Since $\LLS$ is strongly causal we have that the metric topology $\mathcal{D}$ is equivalent to the Alexandrov topology $\mathcal{A}$, thus by Lemma \ref{ftime} (item $(2)$) we have that $f(I^{+}(p) \cap I^{-}(q))=I^{+}(f(p)) \cap I^{-}(f(q))$, therefore, $f$ \cambios{maps} open subsets to open subsets in $(Y,\widetilde{d})$. Now, we have to prove that $f$ is one to one. By contradiction assume that there exist $x,y \in X$ ($x \not = y$) with $f(x)=f(y)$. Take an open neighborhood $U$ of $x$ such that $y \not \in U$   
and consider $V \subset U$ such that any causal curve with endpoints in $V$ is entirely contained in $U$ (see \cite[Lemma 2.38]{KS}), in particular $y \not \in V$. Using the Sequence Lemma \ref{lema:seq} we can take $p,q \in V$ such that $p \ll x \ll q$ and by the defining property on $V$ we have that $I^{+}(p) \cap I^{-}(q) \subset U$, so $y \not \in I^{+}(p) \cap I^{-}(q)$. By Lemma \ref{ftime} again, we have that $f(p) \widetilde{\ll} f(x)=f(y) \widetilde{\ll} f(q)$ and this implies that $p \ll y \ll q$ in contradiction to $y \not  \in V$.    
\end{proof}

Now, we will prove the main result in this section.

\begin{thm}
\label{stronglythm}
	Let $\LLS$ and $\LLSY$ be two Lorentzian length spaces. If $\LLS$ is strongly causal, $(Y,\widetilde{d})$ is locally compact  and $f$ is a surjective distance homothetic map, then $f$ is a homeomorphism and $\LLSY$ is strongly causal. 
\end{thm}

\begin{proof}

	By contradiction assume that $\LLSY$ is not strongly causal. Then, there exists an open subset $U' \in \mathcal{D}$ such that $A' \not \subseteq U'$ for all $A' \in \widetilde{\mathcal{A}}$, where $\widetilde{\mathcal{A}}$ is the Alexandrov topology in $Y$. Take $p' \in U'$ and consider an open neighborhood $B_{\delta_0}^{\widetilde{d}}(p')$ such that $\overline{B_{\delta_0}^{\widetilde{d}}(p')}$ is compact and contained in $U'$. \cambios{Let us consider} a localizing neighborhood $\Omega_{p'}$ with $I^{\pm}(p') \cap \Omega_{p'} \neq \emptyset$. By the Sequence Lemma  there are two sequences $\{r_{n}'\}$ and $\{q_{n}'\}$ in $B_{\delta_0}^{\widetilde{d}}(p')$ converging to $p'$, such that $r_{n}' \widetilde{\ll} p' \widetilde{\ll} q_{n}'$, with $r_n' \widetilde{\ll} r_{n+1}'$ and $q_{n+1}' \widetilde{\ll} q_n'$, $\forall n\in\mathbb{N}$ for all $n$. Therefore, $p' \in I^{+}(r_{n}') \cap I^{-}(q_{n}') \not \subset U'$ and as a consequence $I^{+}(r_{n}') \cap I^{-}(q_{n}') \not \subset \overline{B_{\delta_0}^{\widetilde{d}}(p')}$ for all $n$. Therefore, there exists $w_{n}' \in I^{+}(r_{n}') \cap I^{-}(q_{n}')$ with $w_n' \not \in \overline{B_{\delta_0}^{\widetilde{d}}(p')}$. Since \cambios{$\LLSY$} is causally path connected we have that for each $n$ there exists a future directed causal curve $\gamma_n : [a_n,b_n] \rightarrow Y$ with $\gamma_{n}([a_n,b_n]) \cap \partial B_{\delta_0}^{\widetilde{d}}(p') \neq \emptyset$; this can be seen by taking the future directed timelike curve that goes through $r_{n}',w_{n}'$ and $q_n '$. Take $z_{n}' \in \gamma_{n}([a_n,b_n]) \cap \partial B_{\delta_0}^{\widetilde{d}}(p')$ the last point of entry and note that $r_{n}' \widetilde{\ll} z_{n}' \widetilde{\ll} q_{n}'$.

    The compactness of $\partial B_{\delta_0}^{\widetilde{d}}(p')$ implies that $\{z_{n}'\} \subset \partial B_{\delta_0}^{\widetilde{d}}(p')$ converges to some point $z' \in \partial B_{\delta_0}^{\widetilde{d}}(p')$. The continuity of $f^{-1}:Y \rightarrow X$ implies that the sequences $\{f^{-1}(r_n')\}$ and $\{f^{-1}(q_n')\}$ converge to $p=f^{-1}(p')$ and that $\{f^{-1}(z_n')\}$ converges to $f^{-1}(z')$. Moreover, $f^{-1}(r_n') \ll p \ll f^{-1}(q_n')$ for all $n$. Let us prove that $\{f^{-1}(z_n')\}$ converges to $p$ as well, this will provide a contradiction to the injectivity of $f^{-1}$. Let $U$ be an open neighborhood of $p$, therefore by the strong causality condition of \cambios{$\LLS$} there exists a causally convex neighborhood $V$ of $p$ and contained in $U$. Convergence to $p$ of the sequences $\{f^{-1}(r_n')\}$ and $\{f^{-1}(q_n')\}$ implies that $f^{-1}(r_n')$ and $f^{-1}(q_n')$ are in $V$ for large $n$, so, $f^{-1}(z_n') \in I^{+}(f^{-1}(r_n')) \cap I^{-}(f^{-1}(q_n'))=f^{-1}(I^{+}(r_n') \cap I^{-}(q_n')) \subset V \subset U$ for large $n$, where the first inclusion is due to the causal convexity of $V$. Therefore, $\{f^{-1}(z_n')\}$ also converges to $p=f^{-1}(p')$ and therefore $f^{-1}(z')=p=f^{-1}(p')$ with $p' \neq z'$, in contradiction to the injectivity of $f^{-1}$. We can conclude that $\LLSY$ is strongly causal. Moreover, by applying Lemma \ref{lema:43} to $f^{-1}$ we obtain that $f$ is also continuous, so the distance homothetic map $f:X \rightarrow Y$ is a homeomorphism.  
\end{proof}

As a corollary we have that causal continuity is also induced by distance homothetic maps.

\begin{cor}
Let $\LLS$ and $\LLSY$ be two Lorentzian length spaces. If $\LLS$ is causally continuous, $f$ is a surjective distance homothetic map and $(Y,\widetilde{d})$ is locally compact, then $\LLSY$ is causally continuous. 
\end{cor}

\begin{proof}
	Since $\LLS$ is causally continuous, it is strongly causal. Theorem \ref{stronglythm} implies that $f$ is a homeomorphism and hence $\LLSY$ is also strongly causal. Thus, we only have to show that $\LLSY$ is past and future reflecting. We will prove that $\LLSY$ is past reflecting (the future reflecting case is analogous), this is, let $q' \in Y$  and we will show that $I^{+}(p') \supset I^{+}(q') \Rightarrow I^{-}(p') \subset I^{-}(q')$ for all $p'$. Taking images under the inverse map $f^{-1}$ we have
	\[
	f^{-1}(I^{+}(p') \supset f^{-1}(I^{+}(q'),
	\]
    since $f^{-1}$ is distance homothetic we obtain
	\[
	I^{+}(f^{-1}(p')) \supset I^{+}(f^{-1}(q')). 
	\]
	By hypothesis we have that $\LLS$ is causally continuous, so the previous inclusion implies that $I^{-}(f^{-1}(p')) \subset I^{-}(f^{-1}(q'))$. Therefore, by taking images under $f$ and using Lemma \ref{ftime} we have that 
	\[
		f(I^{-}(f^{-1}(p'))) \subset f(I^{-}(f^{-1}(q'))) \Leftrightarrow  I^{-}(p') \subset I^{-}(q').  
	\]
	Then, $\LLSY$ is past reflecting at $q'$ and a similar reasoning shows that $\LLSY$ is future reflecting at $q'$. Therefore, $\LLSY$ is causally continuous. 
\end{proof}

As we move on to proving that  causal simplicity and global hyperbolicity are preserved under distance homothetic maps, we have to ensure that $f(J^{\pm}(x))\subset J^{\pm}(f(x))$ for all $x \in X$. This will be achieved if we assume the following mild condition on the distance homothetic map:  

\begin{defi}
\label{lcl}
A distance homothetic map $f:\LLS \rightarrow \LLSY$ is called locally causally Lipschitz if for all $x \in X$ there exist an open neighborhood $U \ni x$ and $M>0$ such that $\widetilde{d}(f(x_1),f(x_2)) \leq M d(x_1,x_2)$ for all $x_1 ,x_2 \in U$ with $x_1 \leq x_2$.  
\end{defi}

\begin{prop}
\label{fcausalcurv}
Let $\LLS$ and $\LLSY$ be two Lorentzian length spaces and $f:X \rightarrow Y$ be a surjective homothetic distance map that is locally causally Lipschitz. If $\LLS$ is strongly causal and $(Y,\widetilde{d})$ is locally compact, then $f$ maps causal curves to causal curves. As a consequence, \cambios{$f(J^{\pm}(x))\subset J^{\pm}(f(x))$} for all $x \in X$.
\end{prop}

\begin{proof}
First we have to check that $f \circ \gamma:[a,b] \rightarrow Y$ is  locally Lipschitz continuous with respect to $\widetilde{d}$. In order to do this let $t \in [a,b]$ and consider an open neighborhood $U_{t}$ for $\gamma(t)$ as in Definition \ref{lcl}. Take $t_1, t_{2} \in \gamma^{-1}(U_{t})$ and assume without loss of generality that $t_1 < t_2$, since $\gamma$ is a future directed causal curve we have that $\gamma(t_1) \leq \gamma(t_2)$ in $U_t$. Then, $f$ being a locally causally Lipschitz distance homothetic map on $U_t$ implies that 
\cambios{$$\widetilde{d}(f(\gamma(t_1)),f(\gamma(t_2))) \leq M_t d(\gamma(t_1),\gamma(t_2)) \leq M_t \cdot L |t_2-t_1|,$$} 
where $M_t$ is the Lipschitz constant in $U_{t}$ and $L$ is the Lipschitz constant of $\gamma$. Therefore, $f\circ \gamma:[a,b] \rightarrow Y$ is locally Lipschitz with respect to the metric $\widetilde{d}$. 

Now, we will show that for any future directed causal curve $\gamma:[a,b] \rightarrow X$ the curve $f \circ \gamma :[a,b] \rightarrow Y$ is also future directed causal.  For each $t \in [a,b]$ consider a causally closed neighborhood $(W_t,\widetilde{\le}_{W_{t}})$ for $f(\gamma(t))$. By Theorem \ref{stronglythm}, we know that $\LLSY$ is strongly causal, then there exists a neighborhood $V_{t} \subset W_{t}$ of $f(\gamma(t))$ such that any causal curve with endpoints in $V_{t}$ is entirely contained in $W_{t}$. 
The continuity of $f$ and $\gamma$ implies that $(f\circ \gamma)^{-1}(V_t)$ is an open neighborhood of $t$. Let us prove the following claim: if $t_1,t_2 \in (f\circ \gamma)^{-1}(V_t)$ with $t_1<t_2$, then $f(\gamma(t_1)) \leq f(\gamma(t_2))$.  
To do so, take $t_1<t_2$ in $(f\circ \gamma)^{-1}(V_t)$
and note that $\gamma(t_1) \leq \gamma(t_2)$, since $\gamma$ is a future directed causal curve in $X$. Consider a sequence $\{y_n\} \subset I^{+}(\gamma(t_2))$ with $\gamma(t_1) \ll y_n$, $y_{n} \rightarrow y$ and $y_{n} \in f^{-1}(V_{t})$. By Lemma \ref{ftime} we have that $f(\gamma(t_1)) \widetilde{\ll} f(y_n)$ and observe that $f(\gamma(t_{1})),f(\gamma(t_{2})), f(y_{n}) \in V_{t} \subset W_{t}$. Thus, by the way we chose $V_{t}$ we have that $\gamma(t_{1}) \widetilde{\le}_{W_{t}} f(y_{n})$. Since $f$ is a homeomorphism we have $f(y_{n}) \rightarrow f(\gamma(t_{2}))$. Then, the fact that $\widetilde{\le}_{W_{t}}$ is closed implies that $f(\gamma(t_{1})) \widetilde{\le}_{W_{t}} f(\gamma(t_{2}))$ and thus $f(\gamma(t_{1})) \widetilde{\le} f(\gamma(t_{2}))$.

Consider the open covering $\{(f \circ \gamma)^{-1}(V_t)\}_{t\in [a,b]}$ of $[a,b]$, then there exists a Lebesgue number $\delta > 0$ for this covering. To prove that $f\circ \gamma$ is a causal curve, let $t_1, t_2 \in [a,b]$ with $t_1<t_2$, then we have two cases: 
\begin{enumerate}
\item If $t_2-t_1<\delta$, then, $[t_1,t_2] \subset (f \circ \gamma)^{-1}(V_t)$ for some $t$. Therefore, $f(\gamma(t_1)) \widetilde{\leq} f(\gamma(t_2))$.
\item If $t_{2}-t_1 > \delta$, then, there exists a sequence  $\{t_1,t_1+\delta/2,...,t_1+m_0 \frac{\delta}{2},t_2\}$ with 
$f(\gamma(t_1)) \widetilde{\leq} f(\gamma(t_1+\delta/2)) \widetilde{\leq} ... \widetilde{\leq} f(\gamma(t_1 +m_0 \frac{\delta}{2})) \widetilde{\leq} f(\gamma(t_2))$.   
\end{enumerate}             
In both cases we have $f(\gamma(t_1)) \widetilde{\leq} f(\gamma(t_2))$, hence $f \circ \gamma:[a,b] \rightarrow Y$ is a future directed causal curve. The inclusion $f(J^{+}(x)) \subset J^{+}(f(x))$ for all $x \in X$ readily follows.
\end{proof}

\begin{prop}\label{prop:CSHM}
Let $\LLS$ and $\LLSY$ be two Lorentzian length spaces, and $f$ a surjective homothetic distance map that is locally causally Lipschitz. If $\LLS$ is  causally simple and $(Y,\widetilde{d})$ is locally compact, then $\LLSY$ is causally simple as well. 
\end{prop}

\begin{proof}
By Theorem \ref{stronglythm} we know that $\LLSY$ is strongly causal, so in order to show that $\LLSY$ is causally simple it is enough to prove that $f(J^{\pm}(x))=J^{\pm}(f(x))$ for all $x \in X$. Let $x \in X$ and take $q' \in J^{+}(f(x))$, then there exists a sequence $\{q_{n}'\} \subset I^{+}(f(x))$ with $q_{n}' \rightarrow q'$. Note that $x \ll f^{-1}(q_{n}')$ and $f^{-1}(q_{n}') \rightarrow f^{-1}(q')$ since $f^{-1}$ is a homothetic distance map and a homeomorphism. Observe that $f^{-1}(q') \in \overline{I^+(x)}=J^{+}(x)$ as $\LLS$ is causally simple, so $x \leq f^{-1}(q')$. Thus, $J^{+}(f(x)) \subset f(J^{+}(x))$ and by Proposition \ref{fcausalcurv}  the equality of sets holds. 
From the last equality and the fact that $f$ is a homeomorphism we have that 
\[
J^{\pm}(p')=J^{\pm}(f(f^{-1}(p')))=f(J^{+}(f^{-1}(p')))
\]
and this implies that $J^{\pm}(p')$ is a closed subset since $J^{+}(f^{-1}(p'))$ is a closed subset in $X$ and $f$ a homeomorphism. Then, $\LLSY$ is causally simple.  
\end{proof}

\begin{prop}
Let $\LLS$ and $\LLSY$ be two Lorentzian length spaces, and $f$ a surjective homothetic distance map that is locally causally Lipschitz. If $\LLS$ is globally hyperbolic and $(Y,\widetilde{d})$ is locally compact, then $\LLSY$ is also globally hyperbolic. 
\end{prop}

\begin{proof}
We have that $\LLS$ is globally hyperbolic, then it is causally simple. Therefore, $\LLSY$ is causally simple by Proposition \ref{prop:CSHM}, and thus it is non totally imprisoning as well (see Theorem \ref{causalthm}). Note that by Theorem \ref{stronglythm} we have that $f$ is a homeomorphism. It remains to show that $J^{+}(p') \cap J^{-}(q')$ is compact for all $p',q' \in Y$, so let $p',q' \in Y$ and note that 
\begin{align*}
J^{+}(p') \cap J^{-}(q') &=f(J^{+}(f^{-1}(p'))) \cap f(J^-(f^{-1}(q'))) \\
&=f(J^{+}(f^{-1}(p')) \cap J^-(f^{-1}(q'))),
\end{align*}
where the first equality \cambios{follows} by the proof of Proposition \ref{prop:CSHM}  and the second equality is given by the injectivity of $f$. Moreover, $J^{+}(p') \cap J^{-}(q')$
is compact since $J^{+}(f^{-1}(p')) \cap J^-(f^{-1}(q'))$ is compact in $X$ and $f$ is a homeomorphism. Therefore, $\LLSY$ 
is globally hyperbolic.
\end{proof}

\begin{prop}
Let $\LLS$ and $\LLSY$ be two Lorentzian length spaces and $f$ a surjective homothetic distance map that is locally causally Lipschitz. If $\LLS$ is stably causal, then $\LLSY$ is stably causal as well.
\end{prop}

\begin{proof}
We start by defining a new relation $S^+$ in $Y \times Y$. We will show that $S^+$ is transitive, closed and contains $J^+$, and thus $K^+ \subset S^+$. Let $S^+$ be defined as
\begin{equation*}
S^+ = \{ (f(x),f(y)) \, | \ y \in K^+(x) \}.
\end{equation*}

To show that this relation is transitive, consider $(f(x),f(y))$, $(f(y),f(z)) \in S^+$. Then, $y \in K^+(x)$ and $z \in K^+(y)$. Since $K^+$ in $X$ is transitive, we have $z \in K^+(x)$, and hence $(f(x),f(z)) \in S^+$. 

Now take $(f(x),f(y)) \in J^+$, that is, $f(y) \in J^+(f(x)) \subset \overline{J^+(f(x))}=\overline{I^+(f(x))}$. Let $\{ q_n \} \subset I^+(f(x))$ be a sequence such that $q_n \to f(y)$. Since $f^{-1}$ is a homothetic distance map and a homeomorphism, we have that $x \ll f^{-1}(q_n)$ and $f^{-1}(q_n) \to  y$, and thus $y \in \overline{I^+(x)} = \overline{J^+(x)} \subset K^+(x)$. Therefore, $(f(x),f(y)) \in S^+$.

To show that $S^+$ is closed, let $(f(x),f(y)) \in \overline{S^+}$. Then, there exists a sequence $(f(x_n),f(y_n)) \in S^+$ such that $(f(x_n),f(y_n)) \to (f(x),f(y))$. Since $(f(x_n),f(y_n)) \in S^+$ then $y_n \in K^+(x_n)$, which is equivalent to $(x_n,y_n) \in K^+$. Using the fact that $K^+$ is closed, we conclude that $(x,y) \in K^+$, which in turn implies that $y \in K^+(x)$, that is, $(f(x),f(y)) \in S^+$.

Since $K^+$ is the smallest transitive, closed relation that contains $J^+$, we know that $K^+ \subset S^+$. Note that this implies that $K^+(f(x)) \subset f(K^+(x))$ for all $x \in X$.

In order to show that $K^+$ (in $Y$) is antisymmetric and thus $\LLSY$ is stably causal, let $(f(x),f(y))$ and $(f(y),f(x))$ in $K^+$ (in $Y$). We want to show that $f(x)=f(y)$. Recall that $K^+ \subset S^+$, then it follows that $y \in K^+(x)$ and $x \in K^+(y)$, and therefore $x=y$ since $\LLS$ is stably causal. Thus, $f(x)=f(y)$ and $K^{+}$ is antisymmetric.

\end{proof}

As a final result in this section we have that under some mild assumptions distance homothetic maps send maximal causal curves in $X$ to maximal causal curves in $Y$.

\begin{prop}
Let $\LLS$ and $\LLSY$ be two Lorentzian length spaces and $f$ a surjective homothetic distance map that is locally causally Lipschitz. If $\LLS$ is strongly causal and $(Y,\widetilde{d})$ is locally compact, then $f$ sends maximal causal curves in $X$ to maximal causal curves in $Y$. 
\end{prop}

\begin{proof}
Let $x,y \in X$ with $x \leq y$ and let $\gamma:[a,b] \rightarrow X$ be a maximal future directed causal curve between these points, i.e., $L_{\tau}(\gamma)=\tau(x,y)$.
We want to show that $f \circ \gamma :[a,b] \rightarrow Y$ is also a maximal future directed causal curve. First, note that $f \circ \gamma:[a,b] \rightarrow Y$ is also a future directed causal curve by Proposition \ref{fcausalcurv}. Now, let us prove that $L_{\widetilde{\tau}}(f \circ \gamma)=\widetilde{\tau}(f(x),f(y))$. Recall that by definition,
\[
L_{\widetilde{\tau}}(f \circ \gamma)=\inf \{\sum_{i=0}^{N-1} \widetilde{\tau}(f(\gamma(t_i)),f(\gamma(t_{i+1}))) \mid a=t_0 < t_1 < ... < t_{N}=b\}.
\] 

Since $f$ is a homothetic distance map and $\gamma$ is maximal causal curve we have that

\[
L_{\widetilde{\tau}}(f \circ \gamma)=\inf \{\sum_{i=0}^{N-1} c\tau(\gamma(t_i),\gamma(t_{i+1}))  \mid a=t_0 < t_1 < ... < t_{N}=b\}=c L_{\tau}(\gamma).
\]

Note that $\gamma$ is a maximal future directed causal curve, so $L_{\tau}(\gamma)=\tau(x,y)$. Therefore, $L_{\widetilde{\tau}}(f \circ \gamma)=c \tau(x,y)=\widetilde{\tau}(f(x),f(y))$ and this proves that $f \circ \gamma$ is a maximal causal curve in $\LLSY$.

\end{proof}

\section*{Acknowledgments}
L. Ak\'e was partially supported by FORDECyT, under grant no. 265667. AJ Cabrera Pacheco is grateful to the Carl Zeiss Foundation for its generous support. D. A. Solis acknowledges the kind hospitality of University of T\"ubingen and CIMAT-M\'erida, where part of this work was developed and the financial support of UADY under grant PFCE-2019-12.

\bibliographystyle{amsplain}
\bibliography{LLS-bib}
\vfill

\end{document}